\documentclass{amsart}
\usepackage{graphicx}
\usepackage{epstopdf}
\usepackage{comment}
\usepackage{amsmath,amssymb,amsthm}
\usepackage{thmtools}
\usepackage{hyperref}

\DeclareMathOperator{\Aut}{Aut}
\DeclareMathOperator{\Csp}{CSP}
\DeclareMathOperator{\Qcsp}{QCSP}
\DeclareMathOperator{\mx}{mx}
\DeclareMathOperator{\pp}{pp}
\DeclareMathOperator{\lele}{ll}
\DeclareMathOperator{\mi}{mi}
\DeclareMathOperator{\dpp}{dual-pp}
\DeclareMathOperator{\minop}{min}
\DeclareMathOperator{\mxop}{mx}
\DeclareMathOperator{\qcsp}{QCSP}

\newcommand{\Q}{\mathbb{Q}}

\declaretheorem[numberwithin=section]{theorem}
\declaretheorem[sibling=theorem]{lemma}
\declaretheorem[sibling=theorem]{proposition}

\theoremstyle{definition}

\declaretheorem[sibling=theorem]{question}

\title{Tractability of Quantified Temporal Constraints To The Max}

\begin{document}

\markboth{Bodirsky, Chen, Wrona}
{Quantified Temporal Constraint Satisfaction}

%
%

\author{Manuel Bodirsky}
	\address{Institut f\"{u}r Algebra, TU Dresden,
Germany}
	\email{manuel.bodirsky@tu-dresden.de}
	\urladdr{}
	\thanks{{The research leading to the results presented here has also received funding from the European Research Council under the European Community's Seventh Framework Program (FP7/2007-2013 Grant Agreement no. 257039}}

\author{Hubie Chen}
	\address{Department of Computer Science,
Birkbeck, University of London,
Malet Street,
London WC1E 7HX,
United Kingdom}
	\email{hubie@dcs.bbk.ac.uk}
	\thanks{}

\author{Micha{\l} Wrona}
	\address{
Theoretical Computer Science Department,
 Jagiellonian University,
Poland}
	\email{michal.wrona@uj.edu.pl}
	\urladdr{}
	\thanks{The author is supported by the {\em Swedish Research Council} (VR) under grant 621-2012-3239.}

\maketitle

\renewcommand{\Q}{{\mathbb Q}}
\newcommand{\hcomment}[1]{\marginpar{{\bf Hubie: } #1 }}

\newcommand{\mbcomment}[1]{{\bf Manuel comment: } #1 {\bf End.}}
\newcommand{\mwcomment}[1]{\marginpar{{\bf Micha{\l}: } #1 }}

\newtheorem{observation}[theorem]{Observation}

\begin{abstract}
A \emph{temporal constraint language} is
a set of relations that are first-order definable over $({\mathbb Q};<)$.
We show that several
temporal constraint languages whose constraint satisfaction problem is {maximally tractable} are 
also maximally tractable for the more expressive
\emph{quantified constraint satisfaction problem}. These constraint languages are defined in terms of preservation under certain binary polymorphisms. We also present \emph{syntactic} characterizations of the relations in these languages. 
\end{abstract}


\section{Introduction}
For a fixed structure $\Gamma$ with finite
relational signature $\tau$, the \emph{constraint satisfaction problem
  of~~$\Gamma$},
denoted by $\Csp(\Gamma)$, 
is the problem of deciding whether a given sentence of the form
$\exists x_1 \dots \exists x_n (\psi_1 \wedge \dots \wedge \psi_m)$ is true in $\Gamma$; here, each $\psi_i$ is an atomic $\tau$-sentence, that is, of the form $R(y_1,\dots,y_k)$
for $R \in \tau$ and $y_1,\dots,y_k \in \{x_1,\dots,x_n\}$. 
The given sentence is called an \emph{instance} 
of $\Csp(\Gamma)$, 
and the formulas $\psi_i$ are called the \emph{constraints} 
of the instance. 
The set of relations of $\Gamma$
is called the \emph{constraint language} of $\Gamma$.


A very active research program attempts to classify the computational complexity of $\Csp(\Gamma)$ for all finite structures $\Gamma$;
see e.g.~the collection of survey articles~\cite{CSPSurveys}. 
Particularly interesting are structures $\Gamma$
that are \emph{maximally tractable};
the idea is that those structures $\Gamma$ are the
structures that sit at the frontier between the easy and the
hard CSPs.
The notion of maximal tractability also applies to structures
$\Gamma$ with an \emph{infinite} relational signature. 
We say that $\Gamma$ is \emph{maximally tractable} for the CSP with respect to a class $\mathcal C$ of relations over the same domain (and in this case we also say that the constraint language of $\Gamma$ is \emph{maximally tractable}) if 
\begin{itemize}
\item for every structure $\Gamma'$ 
whose constraint language is finite and contained in the constraint language of $\Gamma$, the problem $\Csp(\Gamma')$ is in P, and
\item for every set of relations ${\mathcal R}\subseteq {\mathcal C}$ that
properly contains the constraint language of $\Gamma$ 
there exists a structure $\Gamma''$ with a finite constraint
language that is contained in $\mathcal R$ such that
$\Csp(\Gamma'')$ is NP-hard.
\end{itemize}

Surprisingly often, when $\Gamma$ is maximally tractable for the CSP with respect to a class of relations,
the more expressive \emph{quantified constraint satisfaction problem for $\Gamma$}, 
denoted by $\Qcsp(\Gamma)$,
 can be solved in polynomial time as well.
 The problem $\Qcsp(\Gamma)$ is defined analogously to $\Csp(\Gamma)$,
 with the difference that, in addition to existential quantification,
universal quantification is also permitted in the input. 
That is, an instance of $\Qcsp(\Gamma)$ is a sentence of the form
$Q_1 v_1 \dots Q_n v_n (\psi_1 \wedge \dots \wedge \psi_m)$, for $Q_1,\dots,Q_n \in \{\exists,\forall\}$, 
to be evaluated in $\Gamma$.

While one might expect that $\Qcsp(\Gamma)$ is typically harder than
$\Csp(\Gamma)$, for the following structures $\Gamma$ 
both $\Csp(\Gamma)$ and
$\Qcsp(\Gamma)$ are polynomial-time decidable: 
\begin{enumerate}
\item 2-element structures $\Gamma$ preserved by $\min$ or by $\max$~\cite{QHorn};
\item 2-element structures $\Gamma$ preserved by the majority operation~\cite{AspvallPlassTarjan};
\item 2-element structures $\Gamma$ preserved by the minority
  operation~\cite{Creignou};
\item structures $\Gamma$ preserved by
a near-unanimity operation~\cite{ChenThesis};
\item structures $\Gamma$ preserved by a Maltsev operation~\cite{qcsp};
\item structures $\Gamma$ that contain a unary relation $U$ for every subset of its domain, and $\Csp(\Gamma)$ is tractable~\cite{RQCSP}. 
\end{enumerate}
See~\cite{QCSP-Survey} for a survey on the state of the art of complexity classification for $\Qcsp(\Gamma)$ for finite structures $\Gamma$. 

Maximal tractability of a structure $\Gamma$ for the QCSP
is defined analogously to maximal tractability of $\Gamma$ for the CSP. Note that if $\Gamma$ is maximally tractable for the CSP with respect to some class of relations $\mathcal C$,
and if for all structures $\Gamma'$ whose constraint language
is finite and contained in the constraint language of $\Gamma$
the problem $\Qcsp(\Gamma')$ is in P as well,
then $\Gamma'$ is maximally tractable for the QCSP with respect to $\mathcal C$ as well. 
This is for instance the case for all the maximally tractable 
structures $\Gamma$ that appear in (1)-(6) above
(each of those 6 items contains structures $\Gamma$ that
are maximally tractable for the CSP with respect to the class of all relations over the same domain). If structures that are
maximally tractable for the CSP
are also maximally tractable for the QCSP, then this is interesting
because it tells us 
that in this case 
the restriction to only existential quantifiers does not make
a hard problem easy.

In this article we study this
transfer of maximal tractability for the CSP
to maximal tractability for the QCSP
for structures $\Gamma$ with an infinite domain.
One of the best-understood classes of infinite structures is the class of all \emph{highly set-transitive} structures; these are the structures $\Gamma$
with the property that for all finite subsets $A$, $B$
of equal size, there exists an automorphism of $\Gamma$ which sends $A$ to $B$. 
By a theorem due to Cameron~\cite{Cameron},
this is the case if and only if $\Gamma$ is isomorphic to a structure with domain ${\mathbb Q}$ whose relations 
are \emph{first-order definable} over $({\mathbb Q};<)$,
i.e., for each relation $R$ of $\Gamma$ 
there is a first-order formula
that defines $R$ over $({\mathbb Q};<)$ (without parameters). 
In the context of constraint satisfaction, the constraint languages of highly set-transitive structures 
 have also been called \emph{temporal constraint languages}, since many problems in qualitative temporal reasoning can be formulated
 as $\Csp(\Gamma)$ for such structures $\Gamma$ (see~\cite{tcsps-journal,ll} for examples). 
We refer to a relation having a first-order definition in 
$({\mathbb Q};<)$ as a \emph{temporal relation}. From now on, maximal tractability will be with respect to the class of all temporal relations over ${\mathbb Q}$.

 The class of temporal constraint languages is very rich, and often provides examples and counterexamples for CSPs of infinite structures $\Gamma$.
 For example, for finite structures $\Gamma$
 it is known that $\Csp(\Gamma)$ can be solved by Datalog if and only
 if the core of $\Gamma$ expanded by constants does not interpret primitively positively the structure
 $({\mathbb Z}_p;\{(x,y,z)\; | \; x - y + z=1\})$; 
this follows from the main result in~\cite{BoundedWidth}
in combination with~\cite{FederVardi}. 
The same statement is false already for the 
simple structure
$({\mathbb Q};\{(x,y,z) \; | \; x>y \vee x>z\})$, which does not interpret  $({\mathbb Z}_p;\{(x,y,z)\; | \; x - y +z=1\})$ primitively positively with constants\footnote{Since it has the polymorphism $f(x,y) = \min(x,y)$ satisfying $\forall x,y (f(x,y)=f(y,x))$; see~\cite{Bodirsky-HDR} for an introduction to primitive positive interpretations and the related universal algebra.},
but which cannot be solved by Datalog; see~\cite{ll}. 
The list of phenomena that are specific to 
infinite-domain constraint satisfaction and that
are exemplified already for temporal constraint 
languages can be prolonged easily; we refer 
to~\cite{Bodirsky-HDR} for more detail.

 The complexity of $\Csp(\Gamma)$ has been completely classified for structures $\Gamma$ with a temporal constraint language (due to~\cite{tcsps-journal}; also see~\cite{Bodirsky-HDR}). 
 In this paper we visit
 some of the CSP maximally tractable temporal constraint 
 languages and study their $\Qcsp$. 
For the languages that we study, we show that 
even the $\Qcsp$ is polynomial-time tractable. 
 Our results are as follows.
\begin{itemize}
\item When all relations of $\Gamma$ are preserved by the operation $(x,y) \mapsto \min(x,y)$, then
$\Qcsp(\Gamma)$ is in P.
\item When all relations in $\Gamma$ are preserved by the operation $\mx$ (introduced in~\cite{tcsps-journal}),
then $\Qcsp(\Gamma)$ is in P.
As with $\min$, the operation $\mx$
 is binary and commutative; we repeat the
definition in Section~\ref{sect:mx}.
\end{itemize}

We complement those results by giving intuitive \emph{syntactic
  descriptions} of the temporal relations that are preserved by
$\min$, and 
likewise of those that are preserved by $\mx$. That is, we present a class of syntactically restricted first-order formulas with the property that a temporal relation is preserved by min if and only if it has a definition by a formula from the class,
and likewise for $\mx$.
Those characterizations are important 
since they give a more explicit description of
the allowed constraints in the respective QCSPs.
A similar syntactic description of relations over finite domains preserved by $\min$ can be found in~\cite{Ordered}. 

The results that we present for temporal constraint languages
preserved by $\min$ also hold in analogous form for those that are
preserved by $\max$. 
The operations $\min$ and $\max$
are \emph{dual} in the following sense: $\max(x,y) = -\min(-x,-y)$. 
Similarly, the operation $\mx$ has the dual operation
$(x,y) \mapsto -\mx(-x,-y)$, and again the results that we have
obtained for $\mx$ also hold analogously for the dual of $\mx$. 

\section{Temporal Constraint Languages}
\label{sect:tcsps}
Let $f \colon D^k \rightarrow D$ be a function.
When $t^1=(t^1_1,\dots,t^1_m),\dots,t^k=(t^k_1,\dots,t^k_m) \in D^m$, 
then $f(t^1,\dots,t^k)$ denotes the tuple
$(f(t_1^1,\dots,t^k_1),\dots,f(t_m^1,\dots,t^k_m))$ from $D^m$ that we obtain by applying $f$ \emph{componentwise}.
A function $f \colon D^k \rightarrow D$ 
\emph{preserves} a relation $R \subseteq D^m$
if for all $t^1,\dots,t^k \in R$ 
it holds that $f(t^1,\dots,t^k) \in R$.
If $f$ preserves all relations of a structure $\Gamma$, we say that $f$ is a \emph{polymorphism} of $\Gamma$.

\vspace{.4cm}
{\bf Example.} The operation $\min \colon {\mathbb Q}^2 \rightarrow {\mathbb Q}$ that maps two rational numbers to the minimum of the two numbers is a polymorphism of the temporal constraint language
$({\mathbb Q}; \leq, <)$, but not of the temporal constraint language
$({\mathbb Q}; \neq)$, since it maps for instance the tuples
$(1,0) \in \mathop{\neq}$ and $(0,1) \in \mathop{\neq}$ to $(0,0) \notin \mathop{\neq}$.

See Figure~\ref{fig:min} for an illustration
of the operation $\min$.
In diagrams for binary operations $f$ as in Figure~\ref{fig:pp_dualpp},
we draw a directed edge from $(a,b)$ to $(a',b')$ if $f(a,b) < f(a',b')$. 
Unoriented lines in rows and columns of picture for an operation $f$ relate pairs of values that get the same value under $f$. 

\begin{figure}[h]
\begin{center}
\includegraphics{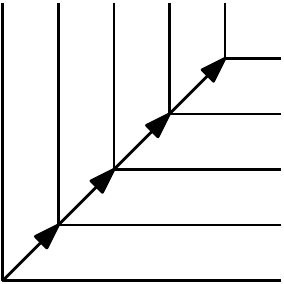}
\caption{Illustration of the operation $\min$.}
\label{fig:min}
\end{center}
\end{figure}

A non-trivial relation that is preserved by $\min$ is the ternary relation defined by the formula $x_1 > x_2 \vee x_1 > x_3$,
and the ternary relation $U$ defined as follows.
\begin{align*}
       U(x,y,z) \; \equiv \; & (x=y \wedge y<z) \\
       \vee \; & (x=z \wedge z<y) \\
       \vee \;  & (x=y \wedge y=z)
\end{align*}

Another important operation for temporal constraint satisfaction
is the binary operation `mx', which is defined as follows~\cite{tcsps-journal}.

$$
\text{mx}(x,y) :=
\left\{\begin{array}{l l l}
\alpha(\min(x,y)) & \text{if } x \neq y \\
\beta(x) & \text{if } x=y
\end{array} \right.
$$
where $\alpha$ and $\beta$ are unary operations that preserve $<$ such that
$\alpha(x) < \beta(x) < \alpha(x+\varepsilon)$ for all $x\in\mathbb{Q}$ and all $0<\varepsilon\in\mathbb{Q}$ (see Figure~\ref{fig:mx}). It is not difficult to show that
such operations $\alpha,\beta$ do exist; see~\cite{tcsps-journal}.

\begin{figure}[h]
\begin{center}
\includegraphics{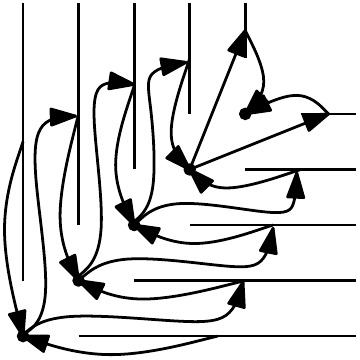}
\caption{Illustration of the operation $\mx$.}
\label{fig:mx}
\end{center}
\end{figure}

Note that mx does not preserve the relation $\leq$.
It also does not preserve the relation $U$ introduced above.
An example of a relation that is preserved by mx is
the ternary relation $X$ defined as follows.

\begin{align*}
       X(x,y,z) \; \equiv \; & (x=y \wedge y<z) \\
       \vee \; & (x=z \wedge z<y) \\
       \vee \; & (y=z \wedge y<x)
\end{align*}

Finally, let $\pp$ be an arbitrary binary operation on $\mathbb{Q}$ such that $\pp(a,b)\leq
\pp(a',b')$ iff one of the following cases applies:
\begin{itemize}
\item $a\leq 0$ and $a \leq a'$
\item $0 < a$, $0 < a'$, and $b\leq b'$.
\end{itemize}
Clearly, such an operation exists. 
For an illustration, see the left diagram in Figure~\ref{fig:pp_dualpp}.
The right diagram of Figure~\ref{fig:pp_dualpp} is  an illustration of the $\dpp$ operation. The name of the operation $\pp$ is derived from the word \emph{`projection-projection'},
since the operation behaves as a projection to the first argument
for negative first argument, and a projection to the second argument for positive first argument.

\begin{figure}[h]
\begin{center}
\includegraphics[scale=.8]{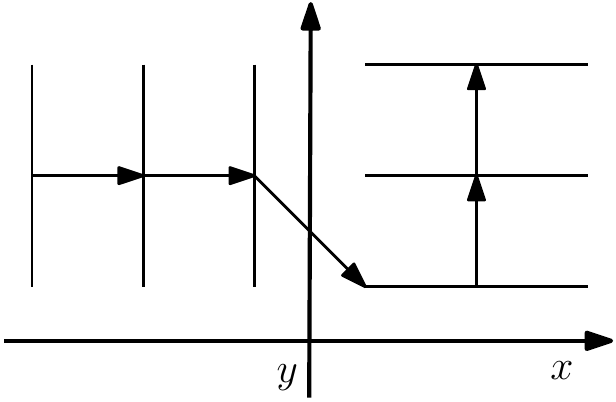}\hskip1cm\includegraphics[scale=.8]{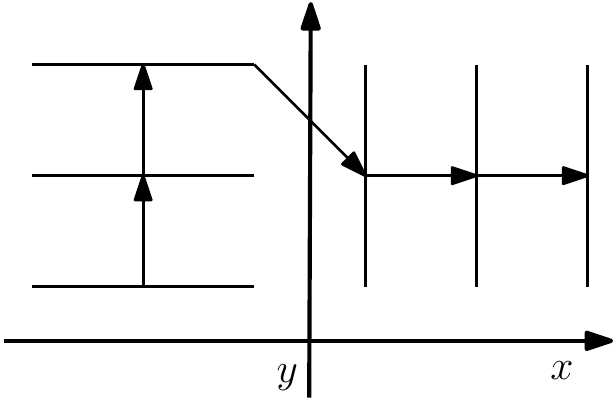}
\caption{A visualization of $\pp$ (left) and $\dpp$ (right).}
\label{fig:pp_dualpp}
\end{center}
\end{figure}

The following two propositions have been shown in~\cite{tcsps-journal}.

\begin{proposition}[\cite{tcsps-journal}]
\label{prop:tractability}
Let $R_1,\dots,R_k$ be temporal relations that are preserved
by min (respectively, mx). Then $\Csp({\mathbb Q}; R_1,\dots,R_k)$ is in P.
\end{proposition}

\begin{proposition}[\cite{tcsps-journal}]
\label{prop:maximality}
Let $R$ be a relation with a first-order definition in $(\mathbb Q;<)$
that is not preserved by min (respectively, mx).
Then there exists a finite set of temporal relations $R_1,\dots,R_k$
preserved by min (respectively, mx)
such that $\Csp({\mathbb Q}; R,R_1,\dots,R_k)$ is NP-hard.
\end{proposition}

A straightforward corollary of Proposition~\ref{prop:tractability} 
and Proposition~\ref{prop:maximality} is that the sets of temporal relations preserved by $\min$ and by $\mx$, respectively, are maximally tractable constraint languages.

The set of temporal relations that 
is preserved by $\min$ is a proper subset of the set of temporal relations that is preserved by $\pp$ (see~\cite{tcsps-journal}); e.g., the relation $\{(x,y,z) \in {\mathbb Q}^3 \; | \; (x = y \wedge y < z) \vee (x > y \wedge y = z)\}$ is preserved by $\pp$, but not by $\min$ and not by $\mx$. Indeed, the CSP for the temporal constraint language that just contains this ternary relation is NP-hard.

\section{Syntactic Characterizations}
It is known that the structure $({\mathbb Q};<)$ has \emph{quantifier-elimination}~\cite{Hodges}, that is, for every first-order formula $\phi$
there exists a quantifier-free formula that is equivalent to $\phi$
over $({\mathbb Q};<)$ (we always allow equality $=$ in our formulas). 
We will call such a quantifier-free formula over the signature $\{<\}$
a \emph{temporal formula}. For the sake of presentation, we write temporal formulas 
using all symbols in $\{ <, >, \leq, \geq, =, \neq \}$. Let  $\phi_R$ be a temporal formula
that defines a temporal relation $R$. For a tuple $t \in R$ and a variable $v$ occurring in $\phi_R$, 
we write $t(v)$ to indicate the value in $t$ corresponding to the variable $v$.

It is also well-known that $({\mathbb Q};<)$ is a \emph{homogeneous structure}~\cite{MacphersonSurvey}, that is,
every isomorphism between finite substructures extends to an automorphism. The set of automorphisms of $(\Q;<)$
will be denoted by $\Aut(\Q;<)$. A set of the form $\{ \alpha(t_1), \ldots, \alpha(t_k) \mid \alpha \in \Aut(\Q; <) \}$ is
called an orbit of $k$-tuples (the orbit of $(t_1, \ldots, t_k)$). 
It is well known that every $k$-ary temporal relation is a union of orbits of $k$-tuples.

Let $\phi_1$ and $\phi_2$ 
be two temporal formulas. We say that $\phi_1$
\emph{entails} $\phi_2$ if and only if $(\Q; <) \models \forall v_1,\dots,v_n 
(\phi_1 \Rightarrow \phi_2)$. 
We also say that an $n$-ary temporal relation \emph{$R$ entails $\phi$} if and only if $(\Q; <) \models \forall v_1,\dots,v_n (\phi_R(v_1, \ldots v_n) \Rightarrow \phi)$, where $\phi_R$ is a temporal formula defining $R$.
We will use the following fact, which follows directly 
from the definition of entailment.

\begin{observation}
\label{obs:replacement}
Let $\phi_R := \Phi \wedge \psi$ be a temporal formula defining a temporal relation $R$ and $\psi'$ a temporal formula entailed by $R$ which entails $\psi$.
Then $\phi'_R := \Phi \wedge \psi'$ also defines $R$.
\end{observation}

\subsection{Characterization of $\pp$-closed Relations}

\begin{theorem}
\label{thm:ppclosed}
A temporal relation $R$ is preserved by $\pp$  if and only if it can be defined as a conjunction of clauses of the form
\begin{equation}
\label{eq:ppclause}
x \neq y_1 \vee \cdots \vee x \neq y_k \vee x \geq z_1 \vee \cdots \vee x \geq z_l,
\end{equation}
where it is permitted that $l = 0$ or $k = 0$, in which cases the clause is a disjunction of disequalities
or contains no disequalities, respectively.
\end{theorem}

To provide the proof we fix $n$, and consider a temporal relation of arity $n$.
All the temporal formulas considered in the proof of Theorem~\ref{thm:ppclosed} are meant to have all variables
in $\{v_1, \ldots, v_n\}$.
In particular we take advantage of temporal formulas of the form:
\begin{equation}
\label{eq:bounds}
\neg(x_1 \circ_1 x_2 \wedge \cdots \wedge x_{k-1} \circ_{k-1} x_{k}),
\end{equation} 
where  $x_1, \ldots, x_k \in \{v_1, \ldots, v_n\}$ are pairwise different
and every $\circ_i$ is in $\{ <, = \}$. We focus on 
formulas of the form~(\ref{eq:bounds}) that satisfy an extra condition that we call the
\emph{unambiguity condition} which states that if  $(v_i = v_j)$ occurs in $\phi$, then $i < j$.
Let $\Phi_R$ be the set containing all temporal formulas $\phi$ over variables $\{v_1, \ldots, v_n\}$ 
of the form~(\ref{eq:bounds}) entailed by $R$ that satisfy the unambiguity condition. 
We will now argue that $\bigwedge \Phi_R$ defines $R$.
Observe first that for every orbit of $n$-tuples in $\Aut(\Q;<)$, there is the negation 
$\phi$ of a formula of the form~(\ref{eq:bounds}) that defines this orbit. 
We can certainly assume that $\phi$ is over variables $\{v_1,\ldots, v_n\}$ and satisfies the  unambiguity 
condition. 
On the other hand, to see that every such $\phi$ defines exactly one orbit 
observe that any two tuples 
$t:=(t_1, \ldots, t_n)$ and $t' := (t'_1, \ldots, t'_n)$ that are in the relation defined by $\phi$  
induce isomorphic substructures in $(\Q;<)$.
By the homogeneity of $(\Q; <)$, an isomorphism between such substructures can be extended to an automorphism
of $(\Q;<)$.  It follows that every orbit may be defined by the negation $\phi$ 
of a formula of the form~(\ref{eq:bounds}).
We are now ready to prove that $R'$ defined by $\bigwedge \Phi_R$ is equal to $R$.
By the definition of $R'$ we have that $R \subseteq R'$.
On the other hand if there is an orbit in $\Aut(\Q;<)$ which is not in $R$, then an appropriate formula of the form~(\ref{eq:bounds})
is entailed by $R$. Thus $R'$ does not contain the orbit as well. Hence we obtain that $\bigwedge \Phi_R$ defines $R$.

Finally, we define an order on $\Phi_R$.
For $\phi_1, \phi_2 \in \Phi_R$ of the form
$\neg(x_1 \circ_1 x_2 \wedge \cdots \wedge x_{k-1} \circ_{k-1} x_{k})$
and 
$\neg(y_1 \circ_1 y_2 \wedge \cdots \wedge y_{l-1} \circ_{l-1} y_{l})$,
respectively,
we will say that $\phi_1$ is less than $\phi_2$, in symbols $\phi_1 \preceq \phi_2$, if $x_1, \ldots, x_k$
is a subsequence of $y_1, \ldots, y_l$ and $\phi_1$ entails $\phi_2$.
It is easy to see that this relation is reflexive and transitive. 
By the definition of $\Phi_R$, we have that two formulas $\phi_1, \phi_2 \in \Phi_R$ 
with the same set of variables entail
each other if and only if they are the same formula.
Hence, $(\Phi_R; \preceq)$ is also antisymmetric, and we have
defined a partial order on $\Phi_R$. Let $\Phi_R^m$ be the set of minimal elements of $(\Phi_R; \preceq)$. 
By Observation~\ref{obs:replacement}, it holds that 
$\bigwedge \Phi^m_R$ defines $R$.

\begin{proof}
To prove the right-to-left implication, it is certainly enough to show that every relation $R_{C}$ defined by a single clause $C$ of the 
form~(\ref{eq:ppclause}) is preserved by $\pp$. 
Assume on the contrary that there are $t_1, t_2 \in R_C$ but
$t_3 = \pp(t_1, t_2)$ is not in $R_C$. It implies that $t_3(x) = t_3(y_1) = \cdots = t_3(y_k)$, and hence by the 
 definition of $\pp$ it holds that $t_1(x) = t_1(y_1) = \cdots = t_1(y_k)$ or $t_2(x) = t_2(y_1) = \cdots = t_2(y_k)$.
In the first case, $t_1(x) \leq 0$ and since for at least one $z_j$ we have that $t_1(z_j) \leq t_1(x)$,
it follows that $t_3(z_j) \leq t_3(x)$. Otherwise, we have that $t_1(x) > 0$.
Again, for at least one $z_j$ it holds that $t_2(z_j) \leq t_2(x)$. It follows that $t_3(z_j) \leq t_3(x)$. Thus, every clause of the 
form~(\ref{eq:ppclause}) is preserved by $\pp$. 

We now show that every temporal relation $R$ preserved by $\pp$ may be defined as a conjunction of clauses of the desired form.
Recall that $\bigwedge \Phi_R^m$ defines $R$. We will obtain the desired definition of $R$ by replacing every 
member $\phi$ of $\Phi_R^m$, which is of the form~(\ref{eq:bounds}), 
with a clause of the form (\ref{eq:ppclause}). If $\phi$ is 
$\neg(x = y_1 \wedge \cdots \wedge y_{k-1} = y_k)$, then we replace it by $(x \neq y_1 \vee \cdots \vee x \neq y_k)$.
Since these formulas are equivalent, 
the transformation does not change the defined relation.

Otherwise, $\phi$ contains at least one occurrence of $<$
and hence we can assume that $\phi$ is of the form:
$$\neg(x = y_1 \wedge \cdots \wedge y_{k-1} = y_k \wedge y_k < z_1 \wedge z_1 \circ_1 z_2 \wedge \cdots \wedge z_{l-1} \circ_{l-1} z_l),$$

\noindent
where every $\circ_i$ is in $\{=, <\}$.
Here we will consider two cases. The first is where 
$R$ contains no tuple $t$ satisfying $t(x) = t(y_1) = \cdots = t(y_k) < t(z_i)$ for all $i \in [l]$. Observe that in this case $R$ entails 
$\psi$ equal to $(x \neq y_1 \vee \cdots \vee x \neq y_k \vee x \geq z_1 \vee \cdots \vee x \geq z_l)$. 
It is easy to verify that $\psi$ entails $\phi$. Hence, by Observation~\ref{obs:replacement}, it
follows that we can safely replace $\phi$ by $\psi$ in $\Phi^m_R$.

If $R$ contains a tuple $t$ satisfying $t(x) = t(y_1) = \cdots = t(y_k) < t(z_i)$ for all $i \in [l]$, then, as we show, $R$ is not preserved by $\pp$.
This leads to a contradiction with the assumption.
Consider a formula $\theta$ equal to $\neg(z_1 \circ_1 z_2 \wedge \cdots \wedge z_{l-1} \circ_{l-1} z_l)$.
We have $\theta \preceq \phi$. Since $\phi$ is in $\Phi^m_R$, it follows that $\theta$ is not
in $\Phi_R$, and hence not entailed by $R$.
It implies that $R$ contains a tuple $t_1$ satisfying $(z_1 \circ_1 z_2 \wedge \cdots \wedge z_{l-1} \circ_{l-1} z_l)$.
Let $\alpha$ be automorphisms of $(\Q;<)$ such that $\alpha(t(x)) < 0 <  \alpha(t(z_i))$ for 
all $i \in [l]$. Then 
$t_2 = \pp(\alpha(t), t_1)$ satisfies $\neg \phi$. Since $\phi$ is entailed by $R$,
it follows $t_2 \notin R$, and thus that $R$ is not preserved by $\pp$.
\end{proof}

\subsection{Characterization of $\min$-closed Relations}

\begin{theorem}
\label{thm:minclosed}
A temporal relation $R$ is preserved by $\min$  if and only if it can be defined as a conjunction of clauses of the form
\begin{equation}
\label{eq:minclause}
x \circ_1 z_1 \vee \cdots \vee x \circ_l z_l,
\end{equation}
where for all $i \in [l]$, we have that $\circ_i$ is in $\{\geq, >\}$.
\end{theorem}

\begin{proof}
Observe that to prove the left-to-right implication, it is enough to 
show that every relation $R_C$ defined by a single clause $C$
of the form~(\ref{eq:minclause}) is preserved by $\min$.
Let $t_1, t_2$ be some tuples in $R$ and $t_3 = \min(t_1, t_2)$.
Then $t_1, t_2$ satisfy $x \circ_i z_i$, and $x \circ_j z_j$ for some $i,j \in [l]$,
respectively. We consider two cases. 
Assume first that $t_1(z_i) = t_2(z_j)$. Here, it is easy to see that
$t_3(x) \geq t_3(z_i)$ and $t_3(x) \geq t_3(z_j)$. Thus if either $\circ_i$
or $\circ_j$ is $\geq$, then we are done. Otherwise, it holds that
$t_1(x) > t_1(z_i)$ and $t_2(x) > t_2(z_j)$. It follows that $t_3(x)$
is greater than both $t_3(z_i)$ and $t_3(z_j)$. This settles the case where
$t_1(z_i) = t_2(z_j)$.
Assume now without loss of generality that 
$t_1(z_i) < t_2(z_j)$. Observe that $t_3(z_i) \leq t_3(x)$.
Thus we are done when $\circ_i$ is $\leq$.
Otherwise we have that $t_1(z_i) < t_1(x)$ and that $t_1(z_i) < t_2(z_j) \leq t_2(x)$.
It follows that $t_3(z_i) < t_3(x)$. It completes the proof of the left-to-right
implication.

We now turn to prove the reverse implication. Assume that there is some relation $R$ preserved by $\min$ that
cannot be defined by a conjunction of clauses of the form~(\ref{eq:minclause}).
By Lemma~23 in~\cite{tcsps-journal}, we have that every relation preserved by $\min$ is also preserved by $\pp$.
Hence, by Theorem~\ref{thm:ppclosed}, the relation $R$ may be defined as a conjunction
of clauses of the form
\begin{equation}
\label{eq:ppclausestrict}
x \neq y_1 \vee \cdots \vee x \neq y_k \vee x \circ_1 z_1 \vee \cdots \vee x \circ_l z_l,
\end{equation}
where $\circ_i \in \{>, \geq \}$ for every $i \in [l]$. 
Consider the set of such definitions of $R$ with a minimal number of disequalities and from this set choose
one with a minimal number of literals. 
Denote that formula by $\phi_R$.

If $\phi_R$ does not have any disequalities, then we are done. Otherwise it
contains a clause $C$ of the form~(\ref{eq:ppclausestrict}) with at least one disequality $(x \neq y_1)$.
Since $\phi_R$ does not contain any unnecessary literals, there is a tuple in $R$ that falsifies all 
literals in $C$ except for $(x \neq y_1)$. Suppose that 
$R$ contains both tuples $t_1, t_2$ that falsify all
literals in $C$ except for $(x \neq y_1)$ and satisfy $(x < y_1)$ and 
$(x > y_1)$, respectively. Let  $\alpha_1$ be an automorphism of $(\Q; <)$
that sends $(t_1(x), t_1(y_1))$ to $(0,1)$, and 
$\alpha_2$ an automorphism of $(\Q; <)$
that sends $(t_2(x), t_2(y_1))$ to $(1,0)$.
Observe that $t_3 = \min(\alpha_1(t_1), \alpha_2(t_2))$ satisfies $(x = y_1)$. 
Since $\min$ preserves 
$=$, $\leq$, and $<$, we have that $t_3$ falsifies all literals in $C$,
hence we have a contradiction with the fact that $R$ is preserved by $\min$.
It follows that $R$ does not have either $t_1$ or $t_2$.
If $R$ does not contain $t_1$, then in $\phi_R$ we can replace $C$ by
$(x > y_1 \vee x \neq y_2 \vee\cdots \vee x \neq y_k \vee x \circ_1 z_1 \vee \cdots \vee x \circ_l z_l)$ 
obtaining a definition of $R$ guaranteed by Theorem~\ref{thm:ppclosed} with a smaller number of
disequalities. From now on we can therefore assume that $R$ contains $t_1$.

Consider now the formula $\phi'_R$ obtained from $\phi_R$ by replacing the clause $C$ 
by the clause
$C' := (y_1 \neq x \vee y_1 \neq y_2 \vee \cdots \vee y_1 \neq y_k \vee y_1 \circ_1 z_1 \vee \cdots \vee y_1 \circ_l z_l)$.
Observe that $C$ and $C'$ entail each other. 
Hence $\phi'_R$ also defines $R$. Consider $C'$ as a part of $\phi'_R$ and observe that
no literal from $C'$ can be removed. Indeed, the new definition would have less  disequlities or the same number
of disequalities and less literals than $\phi_R$. Thus $R$ contains a tuple $t'$ that satisfies 
all disjuncts of $C'$ except for $(y_1 \neq x)$. As in the previous papragraph, we argue that $R$ 
cannot have both tuples $t'_1$, and $t'_2$
that falsify all
literals in $C'$ except for $(x \neq y_1)$ and satisfy $(y_1 < x)$ and 
$(y_1 > x)$, respectively. If $R$ does not contain $t'_1$, then in $\phi'_R$ we can replace $C'$ by
$(y_1 > x \vee y_1 \neq y_2 \vee\cdots \vee y_1 \neq y_k \vee y_1 \circ_1 z_1 \vee \cdots \vee y_1 \circ_l z_l)$ 
obtaining a definition of $R$ in the form guaranteed by Theorem~\ref{thm:ppclosed} with a smaller number of
disequalities. Thus, we can assume that $R$ contains $t'_1$.

Now, suppose towards the contradiction that  $R$ has both: $t_1$ that falsifies all disjuncts of $C$ except for $(x \neq y_1)$
and satisfies $(x < y_1)$; and $t'_1$ that falsifies all disjuncts of $C'$ except for $(y_1 \neq x)$
and satisfies $(y_1 < x)$.
Let $\alpha, \alpha'$ be automorphisms of $\Aut(\Q;<)$
such that $\alpha$ sends $(t_1(x), t_1(y_1))$ to $(0,1)$ and $\alpha'$ sends $(t'_1(x), t'_1(y_1))$ to $(1,0)$.
To complete the proof we will show that $t_4 = \min(\alpha(t_1), \alpha'(t'_1))$ falsifies all disjuncts of $C$.
Since both $t_1$ and $t'_1$ are in $R$, we obtain a contradiction to the assumption that $R$ is preserved by 
$\min$.
Since $\alpha(t_1(x)) = \alpha(t_1(y_2)) = \cdots = \alpha(t_1(y_k)) = 0$ and $\alpha'(t'_1(y_1)) = \alpha'(t'_1(y_2)) = \cdots = 
\alpha'(t'_1(y_k)) = 0$, it follows that
$t_4(x) = t_4(y_1) = \cdots = t_4(y_k) = 0$ and hence $t_4$ falsifies all disequalities in $C$.
Now, the clause $C$ contains a disjunct $(x \circ_i z_i)$ with $\circ_i \in \{>, \geq \}$ and $i \in [k]$ if and only if 
$C'$ contains $(y_1 \circ_i z_i)$. Assume that $\circ_i$ is $>$. The same argument will work for $\geq$.
Observe that $\alpha(t_1)$ satisfies $(x \leq z_i)$ and sends $x$ to $0$;
and that $\alpha'(t'_1)$ satisfies $(y_1 \leq z_i)$ and sends $y_1$ to $0$. It follows that $t_4(x) \leq 0 \leq t_4(z_i)$.
Thus $t_4$ falsifies $(x < z_i)$ and we are done.
\end{proof}

\vspace{.4cm}
{\bf Example.} According to Theorem~\ref{thm:minclosed}, the relation $U$ from Section~\ref{sect:tcsps} 
can be defined by a conjunction of clauses of the form~\ref{eq:minclause}. Indeed,
observe that $((x \geq y \vee x \geq z) \wedge y \geq x \wedge z \geq x)$
defines $U$.

\subsection{Characterization of $\mx$-closed relations}
\label{sect:mx}

Let us say that a boolean relation
$S \subseteq \{ 0, 1 \}^n$
 is \emph{near-affine} if it holds that
$S \cup \{ (1, \ldots, 1) \}$ is preserved by the operation
$a(x, y) = x \oplus y \oplus 1$.
Let us say that a formula $\phi$ 
is \emph{min-affine}
if there exists a near-affine relation $T \subseteq \{ 0, 1 \}^n$
such that $\phi$ is of the form
$$\bigvee_{t \in T} \big (
(\bigwedge_{i, j : t_i = t_j = 0} x_i = x_j) \wedge
(\bigwedge_{i, j : t_i = 0, t_j = 1}  x_i < x_j)
\big )$$
where $x_1, \ldots, x_n$ are variables;
here, the formula $\phi$ is intended to be interpreted over
the ordered rationals. 

\begin{theorem}
\label{thm:mxclosed}
A temporal relation 
$R \subseteq \Q^k$ 
is preserved by $\mx$ if and only if it is
definable by a formula $\phi(v_1, \ldots, v_k)$ that is
the conjunction of min-affine formulas, each of which is over
a subset of $\{ v_1, \ldots, v_k \}$.
\end{theorem}

For a tuple $b = (b_1, \ldots, b_n) \in \Q^n$, 
we define the \emph{min-tuple} of $b$ to be the tuple
$t = (t_1, \ldots, t_n) \in \{ 0, 1 \}^n$ such that
$t_i = 0$ if and only if $b_i$ is the minimum value in $\{ b_1,
\ldots, b_n \}$; note that the min-tuple of a $\Q$-tuple of nontrivial
arity
always
contains at least one entry equal to $0$. Observe also that a tuple $b \in \Q^n$
is in the relation defined by a min-affine formula if and only if the min-tuple of $b$
is in $T$.

We remark that we permit the empty conjunction of min-affine formulas,
which we consider to be true, and we also permit min-affine formulas
where the near-affine relation $T$ is empty, which formulas are
considered to be false.

\begin{proof}
For the right-to-left direction, 
it suffices to show that the claim holds in the case that $\phi$
is defined by a single min-affine formula.
So suppose that $\phi$ is of the form described above,
with respect to $T \subseteq \{ 0, 1 \}^n$.
Let $b = (b_1, \ldots, b_n)$, $b' = (b'_1, \ldots, b'_n)$ be 
two tuples in $\Q^n$ that satisfy $\phi$.
Let $t, t' \in T$ be the min-tuples of $b, b'$, respectively;
so, $t, t'$ 
witness that $b, b'$ satisfy $\phi$.
Let $m, m'$ denote the minimum values in the tuples $b, b'$,
respectively.

We consider two cases. 

Case $m = m'$.
If $t = t'$, 
then the tuple $\mx(b, b')$ is witnessed to satisfy $\phi$
via the tuple $t = t'$; this can be verified by the definition of $\mx$.
If $t \neq t'$, then 
we claim that the tuple
$\mx(b, b')$ 
is witnessed to satisfy $\phi$ via the tuple
$t \oplus t' \oplus 1$.  
First, observe that $t \oplus t' \oplus 1$ is contained in 
$T \cup \{ (1, \ldots, 1) \}$ by the assumption that this relation
is preserved by the operation $a$; since $t \neq t'$, it holds that
$t \oplus t' \oplus 1$ is not equal to $(1, \ldots, 1)$ and is thus
contained in $T$. 
Next, it is straightforwardly verified from the definition of $\mx$
that $\mx(b, b')$ takes on its minimum value at exactly the coordinates
where one of $b, b'$ takes on its minimum value and the other does not,
which are precisely the coordinates where $t \oplus t' \oplus 1$
is equal to $0$.

Case $m \neq m'$.  We assume for the sake of notation
that $m < m'$.  In this case, the minimum entries among all entries
in $b$ and $b'$ are the entries in $b$ that have value $m$.
It follows from the definition of $\mx$ that
$\mx(b, b')$ takes on its minimum value at exactly the coordinates
where $b$ does so, and hence the tuple $\mx(b, b')$ satisfies $\phi$
via the tuple $t \in T$.

For the left-to-right direction, we proceed by induction on the 
arity $k$ of the relation $R \subseteq \Q^k$.
The case $k = 0$ is verified as follows: if $R$ is empty, then take
the empty conjunction;
if $R$ is non-empty, then take the conjunction consisting of 
the single min-affine formula on no variables with $T = \emptyset$.
In what follows, we assume that $k > 0$.
Let $\phi$ be the conjunction of all min-affine formulas 
that are entailed by $R(v_1, \ldots, v_k)$.
Suppose that $b = (b_1, \ldots b_k) \in \Q^k$ satisfies $\phi$.
Our goal is to show $b \in R$.
Let $t \in \{ 0, 1 \}^k$ be the min-tuple of $b$.

We claim that there is a tuple $c \in R$ with min-tuple $t$.
Let $T \in \{ 0, 1 \}^k$ be the set of min-tuples of tuples in $R$.
It follows directly from (Lemma~28 of~\cite{tcsps-journal}) that if
$s, s' \in T$, then $s \oplus s' \oplus 1 \in T \cup \{ (1, \ldots, 1)
\}$.
 From this implication, it is straightforward to verify that $T$ is
 near-affine.
Hence, the conjunction $\phi$ contains as a conjunct 
a min-affine formula $\alpha$ with relation $T$.  
Since $b$ satisfies $\alpha$, it follows that $t \in T$,
from which the claim follows.

If $t = (0, \ldots, 0)$, then $c \in R$ and $b$ are both constant
tuples.
Since $b$ is equal to $c$ under an automorphism of $(\Q; <)$,
we have $b \in R$.
So, we suppose that $t$ contains $1$ as an entry;
for the sake of notation, we assume that $t$ has the form
$(1, \ldots, 1, 0, \ldots, 0)$, where the first $m$ entries are $1$;
here, $0 < m < k$.
We have that $(b_1, \ldots, b_m)$ satisfies all min-affine formulas
on variables from $v_1, \ldots, v_m$ which are entailed by
$R(v_1, \ldots, v_k)$.  
Hence, by induction, it holds that 
$(b_1, \ldots, b_m) \in \pi_{1,  \ldots, m} R$, and that there exists
a tuple of the form
$(b_1, \ldots, b_m, d)$ in $R$, where $d$ is a tuple of length $(k-m)$.
We can apply an automorphism to the tuple
$(b_1, \ldots, b_m, d)$ 
to obtain a tuple
$(b'_1, \ldots, b'_m, d') \in R$ where all entries are positive.
Also, by applying an automorphism to $c \in R$,
we can obtain a tuple of the form
$(c'_1, \ldots, c'_m, 0, \ldots, 0)$ where, for all coordinates $i$
from
$1$ to $m$, it holds that $c'_i > b'_i$.
Applying $\mx$ to the tuples
$(b'_1, \ldots, b'_m, d')$
and
$(c'_1, \ldots, c'_m, 0, \ldots, 0)$,
one obtains the tuple
$(\mx(b'_1, c'_1), \ldots, \mx(b'_m, c'_m),
\mx(d', \overline{0}))$,
where $\overline{0}$ is a tuple of length $(k-m)$ with all entries equal to $0$,
which is equivalent to $b$ under an automorphism.
\end{proof}

\vspace{.4cm}
{\bf Example.} Observe that the relation $X$ closed under $\mx$ presented in Section~\ref{sect:tcsps}
is defined there as a conjunction of  min-affine formulas.

\section{Tractability Results}

In this section, we prove the following theorem.

\begin{theorem}
Let $\Gamma$ be a temporal constraint language
 preserved by an operation 
$h \in \{ \minop, \mxop \}$.
The problem $\qcsp(\Gamma)$ is polynomial-time decidable.
\end{theorem}

Let $\Gamma$ be a relational $\tau$-structure.
Formulas of the form 
$$Q_1 v_1 \dots Q_n v_n (\psi_1 \wedge \dots \wedge \psi_m)$$ where each $\psi_i$ is of the form $R(y_1,\dots,y_k)$
for $R \in \tau$, and where $Q_1,\dots,Q_n \in \{\exists,\forall\}$, will be called \emph{quantified constraint formulas (over $\Gamma)$}. Existentially quantified variables are typically (but not exclusively) denoted by $x,x_1,x_2,\dots$, and
universally quantified variables by $y,y_1,y_2,\dots$.

\begin{lemma}
\label{lemma:eliminate-universal}
Let $\Gamma$ be a temporal constraint language preserved by an operation 
$h \in \{ \minop, \mxop \}$.
Let $\Phi(z_1, \ldots, z_m) = \forall y \exists x_1 \ldots \exists x_n \phi$
be a quantified constraint formula over $\Gamma$ having free variables $\{ z_1, \ldots, z_m \}$.
Let $\Phi'(z_1, \ldots, z_m)$ be the formula
$\exists y \exists x_1 \ldots \exists x_n (\phi \wedge \bigwedge_{i \in [m]} (z_i < y))$.
If the formula $\Phi$ is satisfiable,
then the formulas $\Phi$, $\Phi'$ have the same satisfying assignments over $\Gamma$.
\end{lemma}

\begin{proof}
Observe that for every $q \in \Q$ there exists 
a value $\alpha(q)$ such that 
$h(q, p) = h(p, q) = \alpha(q)$ for all $p > q$.
Notice that $\alpha(q)$ is injective and preserves $<$.

Let $\Psi$ be the formula
$\exists x_1 \ldots \exists x_n \phi$ with free variables
$\{ z_1, \ldots, z_m, y \}$.
Assume that $\Phi$ is satisfiable; then, there exists a tuple
$(g_1, \ldots, g_m) \in \Q^m$ such that, for all $d \in \Q$,
the tuple $(g_1, \ldots, g_m, d)$ satisfies $\Psi$.
Clearly, all satisfying assignments of $\Phi$ are satisfying assignments
of $\Phi'$.  We thus need to show that each satisfying assignment
of $\Phi'$ is a satisfying assignment of $\Phi$.

Let $(f_1, \ldots, f_m) \in \Q^m$ be a satisfying assignment of $\Phi'$.
By the definitions of the formulas, there exists a value $b \in \Q$
with $b > f_i$ for all $i \in  [m]$ such that
$(f_1, \ldots, f_m, b)$ satisfies $\Psi$.
By applying a suitable automorphism to $(g_1, \ldots, g_m)$, we may assume that
$f_i < g_i$ for all $i \in [m]$.
Now consider 
$h((f_1, \ldots, f_m, b), (g_1, \ldots, g_m, d))$
where $d \in \Q$.
In the first $m$ coordinates, this tuple is equal to 
$(\alpha(f_1), \ldots, \alpha(f_m))$.
In the last coordinate, this tuple is equal to $h(b, d)$;
by varying $d$, the value $h(b, d)$ can take on values from every
\emph{orbit},
where here we understand orbit to be with respect to
the 
group consisting of 
the automorphisms of $(\Q, <)$ that fix the values $\alpha(f_i)$.
It follows that $(\alpha(f_1), \dots, \alpha(f_m))$ 
satisfies $\Phi$; since $\alpha^{-1}$ preserves all first-order
definable relations, we conclude that $(f_1, \dots, f_m)$ satisfies $\Phi$.
\end{proof}

With this lemma in hand, we now prove the theorem.
Let $\Gamma$ be a structure that is preserved by
$\minop$ or by $\mxop$.
We now present an algorithm for $\qcsp(\Gamma)$.
For the sake of notation, we assume that the input instance is of the
form
$\forall y_1 \exists x_1 \ldots \forall y_n \exists x_n \phi$,
that is, it exhibits a strict alternation between the two quantifier
types.  The result for the general case can be obtained, for example,
by considering an algorithm that inserts dummy variables/quantifiers into an
arbitrary instance of $\qcsp(\Gamma)$ to massage it into the described
form, and then passes to the algorithm that we describe.
We make use of the fact
 that CSP$(\Gamma)$ can be decided by the algorithms given in~\cite{tcsps-journal};
indeed, those algorithms describe how to compute satisfying
assignments in the event that they exist.

\begin{figure}[h]
\begin{center}
\small
\fbox{
\begin{tabular}{l}
Algorithm for $\qcsp(\Gamma)$ where $\Gamma$ is preserved by min/mx. \\
\\
{\rm Input: an instance
$\forall y_1 \exists x_1 \ldots \forall y_n \exists x_n \phi$
of $\qcsp(\Gamma)$.} \\
\\
Set $\Psi_n(y_1, x_1, \ldots, y_n, x_n) = \phi$.\\
\\
Loop $i = $ $n$ to $1$\\
\hspace{.5cm} Let $\Phi_i = \forall y_i \exists x_i \Psi_i$\\
\hspace{.5cm} Let $\Phi'_i$ be $\exists y_i \exists x_i 
\big(\Psi_i \wedge (x_1 < y_i \wedge \cdots \wedge x_{i-1} < y_i)
\wedge (y_1 < y_i \wedge \cdots \wedge y_{i-1} < y_i))$
\\

\hspace{.5cm} If $\Phi'_i$ is satisfiable\\

\hspace{1cm}Let $w$ be a satisfying assignment\\

\hspace{.5cm} Else\\
\hspace{1cm} Return False\\

\hspace{.5cm} If $w$ does not satisfy $\Phi_i$,\\

\hspace{1cm}  Return False\\

\hspace{.5cm} [Comment: $\Phi_i$, $\Phi'_i$ are equivalent
and both satisfiable]\\

\hspace{.5cm} Let $\Psi_{i-1} = \Phi'_i$\\

End Loop\\
\\
Return True
\end{tabular}}
\end{center}
\end{figure}

We now discuss the correctness of this algorithm.  
The algorithm maintains the invariant that, at the beginning of each loop,
the formula $\Psi_i(y_1, x_1, \ldots, y_i, x_i)$ is equivalent to 
the formula 
$\forall y_{i+1} \exists x_{i+1} \ldots \forall y_n \exists x_n \phi$.
(By equivalent, we mean that the formulas have the same satisfying
assignments over $\Gamma$.)
This is certainly true for $i = n$ by the initialization of $\Psi_n$.

We now consider the behavior of the algorithm for an execution of the
loop.
If the formula $\Phi'_i$ is not satisfiable, it follows that
the formula $\Phi_i$ is not satisfiable
(as any satisfying assignment for $\Phi_i$ is clearly one for
$\Phi'_i$), and hence that the original sentence is false;
the algorithm hence reports ``False'' correctly.

In the case that $\Phi'_i$ is found to be satisfiable,
consider first the case that the satisfying assignment $w$ for $\Phi'_i$ does not
satisfy $\Phi_i$. It follows 
by Lemma~\ref{lemma:eliminate-universal}
that $\Phi_i$ is not satisfiable, and hence that the original sentence is
false; the algorithm hence reports ``False'' correctly.

In the case that the assignment $w$ does satisfy $\Phi_i$, 
by Lemma~\ref{lemma:eliminate-universal},
the formulas $\Phi_i$, $\Phi'_i$ are equivalent.
Consequently, when $\Psi_{i-1}$ is set equal to $\Phi'_i$,
by the definition of $\Phi_i$ and the fact that the
invariant held at the beginning of the loop, the invariant 
will hold for the next execution of the loop.
In the case that $i = 1$, if the loop does not return false,
by the given argumentation, we have that both $\Phi_1$ and $\Phi'_1$
are true, and hence that the original sentence is true.

\section{Discussion}
The polynomial-time algorithms for the QCSP of temporal constraint languages that are preserved by $\min$ or preserved by $\mx$ that we have seen in this paper
properly strengthen previously known algorithmic results
for the QCSP:
\begin{itemize}
\item In~\cite{CharatonikWronaLPAR}, it has been shown that instances of the QCSP where all constraints are of the form $x \geq y_1 \vee \cdots \vee x \geq y_k$ can be solved in polynomial time. All such constraints are preserved by $\min$  (this follows immediately from the syntactic characterization that we gave for temporal relations that are preserved by $\min$). 
\item In~\cite{collaps}, it has been shown that 
the QCSP for all temporal constraint languages preserved by $\min$ is in NP.
\end{itemize}

There are (up to duality) two more 
temporal constraint languages whose CSP is maximally tractable~\cite{tcsps-journal}.
One of them is characterized by a certain binary polymorphism called $\lele$. The tractability for the CSP for $\lele$-closed languages has been shown in~\cite{ll}. The
QCSP for $\lele$-closed languages is coNP-hard~\cite{qecsps}, so 
this is an instance where maximal tractability doesn't transfer to the
QCSP.
There is the temporal constraint language that contains all relations preserved by the binary operation $\mi$ (for a definition, see~\cite{tcsps-journal}). 
In this case, the tractability of the QCSP is open.
In fact, already the following is open.

\begin{question}
Is $\Qcsp(({\mathbb Q};\{(x,y,z) \; | \; x = y \Rightarrow y \leq z\}))$ in P?
\end{question}
The relation defined by $x = y \Rightarrow y \leq z$ is indeed
preserved by $\mi$ (see~\cite{tcsps-journal}).
Finally, the temporal languages preserved by a constant
operation have a tractable CSP~\cite{tcsps-journal}.
As shown in previous work~\cite{CharatonikWronaLPAR,CharatonikWronaCSL}, the QCSP for such languages may be
contained in LOGSPACE, be NLOGSPACE-complete,
be P-complete, be NP-complete, or be PSPACE-complete.

Another field of open problems is related to the question
how bad the complexity of the QCSP becomes when
we are outside the maximal tractable classes. From maximal
tractability of the CSP we only obtain NP-hardness, but
very often the problem might become even PSPACE-hard.
Very often, questions in this context lead to the following
notorious open problem.

\begin{question}
Is $\Qcsp(({\mathbb Q};\{(x,y,z) \; | \; x = y \Rightarrow y = z\}))$ PSPACE-complete?
\end{question}

We only know that this QCSP is coNP-hard~\cite{qecsps}. 

Postprint note: D. Zhuk and B. Martin solved Question 5.2~\cite{comp-eqcsps}.

\bibliographystyle{abbrv}
\bibliography{local}

\end{document}